\documentclass[10pt,conference]{IEEEtran}

\usepackage[dvips]{color}
\usepackage{tabu}
\usepackage{comment}
\usepackage{todonotes}
\usepackage{epsf}
\usepackage{epsfig}
\usepackage{times}
\usepackage{epsfig}
\usepackage{graphicx}
\usepackage{bbold}
\usepackage{mathtools}
\usepackage{mathrsfs}
\usepackage{amssymb,amsmath}
\usepackage{pdfpages}
\usepackage{epstopdf}
\usepackage{dsfont}
\usepackage{lettrine} 
\usepackage{amsmath,epsfig,amssymb,algorithm,algpseudocode,amsthm,cite,url}

\usepackage[justification=centering]{caption}
\usepackage{subcaption}
\usepackage{bbm}
\usepackage{algorithm}% http://ctan.org/pkg/algorithm
\usepackage{algpseudocode}
\usepackage{algpseudocode}% http://ctan.org/pkg/algorithmicx
\algnewcommand{\LineComment}[1]{\State \(\triangleright\) #1}
\usepackage{csquotes}
\usepackage{verbatim}
\usepackage{arcsaops}
\usepackage[english]{babel}
\usepackage{amsmath,amssymb}

\usepackage{verbatim}

\newtheorem{theorem}{\bf Theorem}
\let\oldtheorem\theorem
\renewcommand{\theorem}{\oldtheorem\normalfont}

\newtheorem{proposition}{\bf Proposition}
\let\oldproposition\proposition
\renewcommand{\proposition}{\oldproposition\normalfont}
\newtheorem{lemma}{\bf Lemma}
\let\oldlemma\lemma
\renewcommand{\lemma}{\oldlemma\normalfont}

\newtheorem{definition}{\bf Definition}
\let\olddefinition\definition
\renewcommand{\definition}{\olddefinition\normalfont}
\newtheorem{remark}{Remark}
\let\oldremark\remark
\renewcommand{\remark}{\oldremark\normalfont}

\usepackage{setspace}	% Remove in double column version. Also search for \setstretch in the body of the paper and comment these commands for double column
\allowdisplaybreaks
\begin{document}
% paper title
\title{\huge  Mobility Management for Heterogeneous Networks: Leveraging Millimeter Wave for Seamless Handover}\vspace{0em}
\author{
\authorblockN{Omid Semiari$^{1}$, Walid Saad$^{2}$, Mehdi Bennis$^3$, and Behrouz Maham$^4$}\\\vspace*{-.5em}
\authorblockA{\small $^{1}$\textcolor{black}{Department of Electrical Engineering, Georgia Southern University, Statesboro, GA, USA, Email: \protect\url{osemiari@georgiasouthern.edu}}\\
	$^{2}$Wireless@VT, Bradley Department of Electrical and Computer Engineering, Virginia Tech, Blacksburg, VA, USA, Email: \protect\url{walids@vt.edu}\\
\small $^{3}$ Centre for Wireless Communications, University of Oulu, Finland, Email: \url{bennis@ee.oulu.fi}\\
$^{4}$Department of Electrical and Electronic Engineering,
Nazarbayev University, Astana, Kazakhstan, Email: \url{behrouz.maham@nu.edu.kz}
}\vspace{-2em}
    \thanks{This research was supported by the U.S. National Science Foundation under Grants CNS-1460316, CNS-1513697, IIS-1633363, and the Academy of Finland CARMA project.}%
  }
\IEEEoverridecommandlockouts
%
% make the title area
\maketitle
\vspace{-1em}
\begin{abstract}
One of the most promising approaches to overcome the uncertainty and dynamic channel variations of millimeter wave (mmW) communications is to deploy dual-mode base stations that integrate both mmW and microwave ($\mu$W) frequencies. In particular, if properly designed, such dual-mode base stations can enhance mobility and handover in highly mobile wireless environments. In this paper, a novel approach for analyzing and managing mobility in joint $\mu$W-mmW networks is proposed. The proposed approach leverages device-level caching along with the capabilities of dual-mode base stations to minimize handover failures and provide seamless mobility. First, fundamental results on the caching capabilities, including caching probability and cache duration, are derived for the proposed dual-mode network scenario. Second, the average achievable rate of caching is derived for mobile users. Then, the impact of caching on the number of handovers (HOs) and the average handover failure (HOF) is analyzed. The derived analytical results suggest that content caching will reduce the HOF and enhance the mobility management in heterogeneous wireless networks with mmW capabilities. Numerical results corroborate the analytical derivations and show that the proposed solution provides significant reductions in the average HOF, reaching up to $45\%$, for mobile users moving with relatively high speeds. \vspace{-0.1cm}
\end{abstract}
\section{Introduction} \label{intro}\vspace{-0cm}

The proliferation of bandwidth-intensive wireless applications such as social networking, high definition video streaming, and mobile TV have drastically strained the capacity of wireless cellular networks. To cope with this traffic increase, several new technologies are anticipated for 5G cellular systems, including: 1) dense deployment of small cell base stations (SBSs), and 2) leveraging the large amount of available bandwidth at \emph{millimeter wave (mmW)} frequencies \cite{Ghosh14}. In fact, SBSs will boost the capacity of wireless networks by reducing the cell sizes and removing the coverage holes. Meanwhile, mmW communications will provide high data rates, by exploiting directional antennas and transmitting over a large bandwidth that can reach up to $5$ GHz. 

However, one of the key practical issues associated with the dense deployment of SBSs is frequent handovers (HOs) which increases the overhead and delay in heterogeneous networks (HetNets). In addition, handover failure (HOF) is more common in HetNets, particularly, for \emph{mobile user equipments (MUEs)} with higher speeds \cite{6384454}. In fact, due to the small and disparate cell sizes in HetNets, MUEs will not be able to successfully finish the HO process by the time they trigger HO and pass a target SBS. 
%In particular, mobile user equipments (MUEs) will ...

To enhance mobility management in  HetNets, an extensive body of work has appeared in the literature \cite{Mulle,6654905,1325888,6587998,7247509,Khan1,7562411,7565107,7354528}. \textcolor{black}{One widely adopted scheme to cope with channel quality variations in mobile environment is the dynamic adaptive streaming over HTTP (DASH) protocol. The authors in \cite{Mulle} analyze the performance of DASH in vehicular networks. In \cite{6654905}, a novel multi-user DASH protocol is proposed to enhance quality of experience for MUEs.} \textcolor{black}{ However, the DASH protocol offers content segments with different quality (different data rate) and it selects the content with lower quality when throughput decreases which results in a lower quality-of-service (QoS). Therefore, despite its importance, the DASH protocol alone will not be sufficient to meet the stringent requirements of 5G applications such as HD TV, uncompressed video streaming, or virtual reality (VR).}

 In \cite{1325888}, the authors comprehensively survey mobility management in IP networks. 
%The authors in \cite{6563279} survey HO decision algorithms that focus on improving HO between femtocells and LTE-Advanced systems. 
The work presented in \cite{6587998} overviews existing approaches for vertical handover decisions in HetNets. In \cite{7247509}, the authors analyze the impact of channel fading on mobility management in HetNets. %In addition, the work in \cite{7247509} shows that increasing the sampling period for HO decision decreases the fading impact, while increasing the ping-pong effect. 
The work in \cite{Khan1} introduces an HO scheme that considers the speed of MUEs to decrease frequent HOs in HetNets. Moreover, the authors in \cite{7562411} propose an HO scheme that supports soft HO by allowing MUEs to connect with both a macrocell base station (MBS) and SBSs. Furthermore, a distributed mobility management framework is proposed in \cite{7565107} which uses multiple frequency bands to decouple the data and control planes. In \cite{7354528}, an HO scheme for mmW networks is proposed in which the MBS acts as an anchor for mmW SBSs to manage control signals. Although interesting, the works in \cite{1325888,6587998,7247509,Khan1,7562411,7565107} do not consider mmW communications. 
%In addition, it does not study the opportunities that caching techniques can provide for mobility management. 
 Moreover, \cite{7354528} assumes that line-of-sight (LoS) mmW links are always available and provides no analytical results to capture the directional nature of mmW communications. 
 %In \cite{7118239}, the authors propose a resource allocation scheme for hybrid mmW-$\mu$W networks that enhances video streaming by buffering content over mmW links. However, \cite{7118239} does not address any mobility management challenge, such as frequent HOs or HOF. 

The main contribution of this paper is a novel mobility management framework that addresses critical handover issues, including frequent HOs and HOF in emerging dense wireless cellular networks with mmW capabilities. In fact, we adopt a model that allows MUEs to cache their requested content by exploiting high capacity mmW connectivity whenever available. Thus, the MUEs will be able to use the cached content and avoid performing any HO, while passing SBSs with relatively small cell sizes. In addition, we propose a geometric model to  derive tractable, closed-form expressions for key performance metrics, including the probability of caching, cumulative distribution function of caching duration, and the average data rate for caching at an MUE over a mmW link. Moreover, we provide insight on the achievable gains for reducing the number of HOs and the average HOF, by leveraging caching in mmW-$\mu$W networks. Under practical settings, we show that HOF can be decreased by up to $45 \%$, even for MUEs moving at high speeds.  

%traffic offloads at $\mu$W frequency band by leveraging mmW resources. Compared to CSI-based scheduling being oblivious to the per application delay tolerance, the proposed approach achieves substantial gains in terms of enhanced QoS. 

The rest of this paper is organized as follows. Section II presents the system model. Section III presents the analysis for caching in mobility management. Performance analysis are provided in Section IV. Section V presents the simulation results and Section VI concludes the paper.

\section{System Model}
Consider a HetNet composed of an MBS and $K$ SBSs within a set $\mathcal{K}$, distributed uniformly across an area. 
%An arbitrary SBS $k \in \mathcal{K}$ can be classified as a picocell or a femtocell, depending on its transmit power $p_k$. Picocells are typically deployed in outdoors, while femtocells are relatively low-power and suitable for indoor deployments. 
%The SBSs operate at a different microwave ($\mu$W) frequency band than the MBS to avoid interfering with the MBS. In addition, 
The SBSs are equipped with mmW front-ends to serve MUEs over either mmW or $\mu$W frequency bands \cite{7929424}. The dual-mode capability allows to integrate mmW and $\mu$W radio access technologies (RATs) at the medium access control (MAC) layer of the air interface and reduce the delay and overhead for fast vertical handovers between both RATs. %c\cite{7503786}. 
We consider a set $\mathcal{U}$ of $U$ MUEs that are distributed randomly and move across the considered geographical area during a time frame $T$. Each user $u \in \mathcal{U}$ moves in a  random direction $\theta_u \in \left[0,2\pi \right]$, with respect to the $\theta=0$ horizontal angle, which is assumed fixed for each MUE over a considered time frame $T$. In addition, we consider that an MUE $u$ moves with an  average speed $v_u \in \left[v_{\text{min}},v_{\text{max}}\right]$. The MUEs can receive their requested traffic over either the mmW or $\mu$W band.

%The considered HetNet is overlaid with $M$ millimeter wave SBSs (mmW-SBSs), within a set $\mathcal{M}$, distributed randomly to serve the mobile users, once a mmW link is feasible. 

%As elaborated in Section \ref{intro}, mmW signals have poor penetration and diffraction characteristics. Therefore, SBSs may not provide a ubiquitous coverage for all associated MUEs. However, high capacity links can be achieved, if the transceivers' antenna beams at the SBS and the MUE are directed toward one another. Therefore, the proposed dual-mode architecture for SBSs will celebrate the available bandwidth at mmW frequencies, while guaranteeing a certain QoS for MUEs. 

%Moreover, due to operating at different frequency bands, mmW-SBSs can substantially improve the rate at coverage holes and at the cell edge, without interfering with SBSs operating at microwave frequencies. To this end, we focus on the mmW-SBS deployments at the cell edges of SBSs in a HetNet\footnote{Considering mmW-SBSs at other locations is complementary to our proposed model and will not affect the analysis provided in this work.}.
\vspace{-.1cm}
\subsection{Channel model}
The large-scale channel effect over mmW frequencies for a link between an SBS $k$ and an MUE $u \in \mathcal{U}$, in dB, is\footnote{The free space path loss model in \eqref{pathloss_mmw} has been adopted in many existing works, such as in \cite{Ghosh14} that carry out real-world measurements to characterize mmW large scale channel effects.}
	\begin{align}\label{pathloss_mmw}
	L(u,k)=20\log_{10}\left(\frac{4\pi r_0}{\lambda}\right)
	\!+\! 10\alpha \log_{10}\left(\frac{r_{u,k}}{r_0}\right)\!+\!\chi, 
	\end{align} 
where \eqref{pathloss_mmw} holds for $r_{u,k}\geq r_0$, with $r_0$ and $r_{u,k}$  denoting, respectively, the reference distance and distance between the MUE $u$ and SBS $k$. In addition, $\alpha$ is the path loss exponent, $\lambda$ is the wavelength at carrier frequency $f_c = 73$ GHz over the E-band, and $\chi$ is a Gaussian random variable with zero mean and variance $\xi^2$. The path loss parameters $\alpha$ and $\xi$ will naturally have different values, depending on whether the mmW link is LoS or non-LoS (NLoS). Over the $\mu$W frequency band, the path loss model follows \eqref{pathloss_mmw}, however, with parameters that are specific to sub-6 GHz frequencies. 

An illustration of the considered HeNet is shown in Fig. \ref{Model1}. The coverage for each SBS at the $\mu$W frequency is shown based on the maximum received signal strength (max-RSS) criteria. In addition, white spaces in Fig. \ref{Model1} show the areas that are covered solely by the MBS. Here, we observe that shadowing effect can adversely increase the ping-pong effect for MUEs. To cope with this issue, the 3GPP standard suggests L1/L3 filtering which basically applies averaging to RSS samples, as detailed in \cite{7247509}.  
\begin{figure}[!t]
	\centering
	\centerline{\includegraphics[width=7cm]{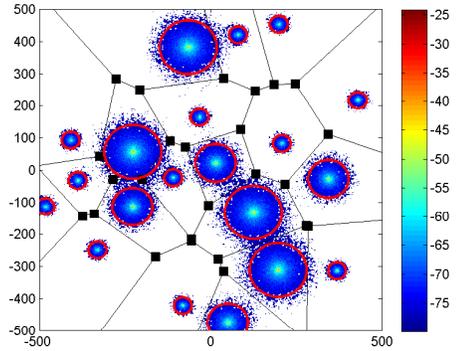}}\vspace{-.5cm}
	\caption{\footnotesize SBSs coverage with RSS threshold of $-80$ dB. Red circles show the simplified cell boundaries.}\vspace{-0.2cm}
	\label{Model1}
\end{figure}\vspace{0em}
\subsection{Antenna model and configuration}
To overcome the excessive path loss at $f_c$, the MUEs are equipped with electronically steerable antennas which allow them to achieve beamforming gains at a desired direction. The antenna gain pattern for MUEs follows the simple and widely-adopted sectorized pattern which is given by: %\cite{7110547}:
\begin{align}\label{gainMUE}
G(\theta)=\begin{cases}
G_{\text{max}}, &\text{if} \,\,\,\,\,\theta <|\theta_m|,\\
G_{\text{min}}, &\text{otherwise},
\end{cases}
\end{align}
where $\theta$ and $\theta_m$ denote, respectively, the azimuth angle and the antennas' main lobe beamwidth. Moreover, $G_{\text{max}}$ and $G_{\text{min}}$ denote, respectively, the antenna gain of the main lobe and side lobes. For MUEs, we use a model similar to the  sectorized pattern in \eqref{gainMUE}. However we allow each SBS to form $N$ beams, either by using $N$ antenna arrays or forming multi-beam beamforming. The beam patten configuration of an MBS is shown in Fig. \ref{model2}, where $N$ equidistant beams in $\theta \in \left[0, 2\pi \right]$ are formed. To avoid the complexity and overhead of beam-tracking for mobile users, the location of the SBSs' beams in azimuth is fixed. In fact, an MUE can connect to an SBS over a mmW link, if the MUE traverses the area covered by the mmW beams of this SBS.  It is assumed that for a desired link between an SBS $k$ and an MUE $u$, the overall transmit-receive gain is $\psi_{u,k}=G_{\text{max}}^2$. 
\vspace{-.1cm}
\subsection{Traffic model}\vspace{-.1cm}
Video streaming is one of the wireless services with most stringent QoS requirement. The QoS achieved by such services is severely impacted by the delay that results from frequent handovers in HetNets. In addition, HOFs can significantly degrade the performance by making frequent service interruptions. Therefore, our goal is to enhance mobility management for MUEs that request video traffic. Each video content is partitioned into small segments, each of size $B$ bits. \textcolor{black}{The network incorporates caching \cite{7875131,7412759} over the mmW frequency band to store incoming video segments at the MUE's storage, whenever a high capacity mmW connection is available.} Considering a cache size of $\Psi_u$ for an arbitrary MUE $u$, associated with an SBS $k$, the number of video segments that can be cached at MUE $u$ will be
\begin{align}\label{traffic1}
	M^c(u,k) = \min\Big\lbrace\bigg\lfloor \frac{\bar{R}^c(u,k)t_u^c}{B}\bigg\rfloor, \frac{\Psi_u}{B}\Big\rbrace,
\end{align}
where $\lfloor . \rfloor$ and $\min\lbrace . ,. \rbrace$ denote, respectively, the floor and minimum operands. In addition, $t_u^c$ is the \emph{caching duration} which is equal to the time needed for an MUE $u$ to traverse a mmW beam at its serving SBS. Considering the small green triangle in Fig. \ref{model2} as the current location of an MUE crossing a mmW beam, caching duration will be $t_u^c=r_u^c/v_u$. Moreover, $\bar{R}^c(u,k)$ is the \emph{average achievable rate} for the MUE $u$ during $t_u^c$. Given $M^c(u,k)$ and the  video play rate of $Q$, specified for each video content, the distance an MUE $u$ can traverse with speed $v_u$, while playing the cached video content will be 
\begin{align}\label{traffic2}
d^c(u,k) = \frac{M^c(u,k)}{Q}v_u.
\end{align}
In fact, the MUE can traverse a distance $d^c(u,k)$ by using the cached video content without requiring an HO to any of the target cells. Meanwhile, the location information and control signals, such as paging, can be handled by the MBS over the $\mu$W frequency band. \textcolor{black}{ Here, we note that buffering segments which is conventionally used in sub-6 GHz systems can also be employed within the proposed HO muting scheme. However, by using buffering, the network can only download a limited number of segments in advance. In contrast, with caching, an MUE can store an entire or a large portion of a content. Thus, by using caching and not just buffering, a mobile MUE will be able to traverse a longer distance $d^c$, while playing the stored content. 
%Furthermore, existing sub-6 GHz cellular networks employ buffering, since fast caching of video content is not always applicable over the bandwidth-limited $\mu$W frequencies.
 Furthermore, by using the mmW frequency band, one can leverage the substantial amount of bandwidth which is suitable to  realize the idea of  fully-fledged content caching in practice.}
\vspace{-.1cm}
\subsection{Handover procedure and performance metrics}\vspace{-.0cm}
The HO process in the 3GPP standard proceeds as follows \cite{3gpp}: 1) Each MUE will do a cell search every $T_s$ seconds that can be configured by the network or directly by MUEs, 2) If any target cell offers an RSS plus a hysteresis that is higher than the serving cell, even after L1/L3 filtering of input RSS samples, the MUE will wait for a time-to-trigger (TTT) of $\Delta T$ seconds to measure the average RSS from the target cell, 
\begin{comment}
as follows:
\begin{align}\label{trigger}
\frac{1}{v_u \Delta T}	\int P_r(u,k,\boldsymbol{x}_{u})d\boldsymbol{x}_{u} < P_{r,th}(u),
\end{align} 
where the integral is over the path that MT $u$ traverse within TTT duration $\Delta T$. In \eqref{trigger}, $P_r(u,k,\boldsymbol{x}_{u})$ is the RSS received by the MUE $u$ from SBS $k$, while being located at $\boldsymbol{x}_u$. In addition, $P_{r,th}(u)$ denotes the threshold RSS to satisfy the QoS of MUE $u$'s traffic.
\end{comment} 
3) If the average RSS is higher than that of the target SBS during TTT, the MUE triggers HO and sends a measurement report to its serving cell. 
The averaging over the TTT duration will reduce the ping-pong effect resulting from instantaneous channel state information (CSI) variations, and 4) HO will be executed after the serving SBS sends HO information to the target SBS.

In our model, we modify the above HO procedure to leverage the caching capabilities of the MUEs during mobility. Here, we let each MUE dynamically determine $T_s$, depending on the number of remaining video segments $M$ in the cache, the video play rate, and the MUE's speed. That is, an MUE $u$ will mute the cell search while $M/Q$ is greater than $\Delta T$. In this way, the MUE will have $\Delta T$ seconds to search for a target SBS before the cached content runs out. 
%As we show later in details, the proposed HO procedure with cell search muting strategy will provide significant enhance the mobility management.

Next, we consider the HOF as a key performance metric for an HO procedure. One of the main reasons for the potential increase in HOF is due to the relatively small cell sizes in HetNets, compared to the MBS coverage. In fact, HOF occurs if the time-of-stay (ToS) for an MUE is less than the minimum ToS (MTS) required for performing a successful HO. That is,
\begin{align}\label{HOF1}
	\gamma_{\text{HOF}}(u,k) = \begin{cases}
	1, \,\,\,\,\,\, \text{if}\,\,\, t_{u,k}<t_{\text{MTS}},\\
	0, \,\,\,\,\,\, \text{otherwise},
	\end{cases}
\end{align}
where $t_{u,k}$ is the ToS for MUE $u$ to pass across SBS $k$ coverage. Although a short ToS may not be the only cause for an HOF, it becomes  critical for small cell setting and MUEs with high speeds \cite{3gpp}.
\begin{comment}
Moreover, MUEs have to spend a total energy $E^s$ for searching the $\mu$W carrier for synchronization signal, and decoding the broadcast channel (system information)
of the detected cells  \cite{6515049}. Hence, the total consumed energy by an MUE for cell search during time $T$  will be
\begin{align}\label{P1}
E^s = P^s  \frac{T}{T_s},
\end{align}
where $P^s$ is the consumed power per each cell search.
%where $M = \big\lfloor \frac{|\mathcal{D}|}{v_u T_s}\big\rfloor $ is the total scanning periods and $N^s(i, \boldsymbol{x}_u)$ is the number of scans at $i$-th scan period. Here, we note that $N^s$ depends on the location of MUE at the time of scanning. In addition, 

Here, we note that increasing $T_s$ reduces $E^s$ which is desirable. However, less frequent scans is equivalent to less HOs to SBSs. Therefore, there is a tradeoff in reducing the consumed power for cell search and maximizing traffic offloads from the MBS to SBSs. Content caching will allow to increase $T_s$, while maintaining traffic offloads from the MBS.
\end{comment} 
Next, we adopt a geometric framework to analyze the caching opportunities, in terms of the caching duration $t^c$, and the average achievable rate $\bar{R^c}$, for MUEs moving at random directions in joint mmW-$\mu$W HetNets. \vspace{-.0cm}
%%%%%%%%%%%%%%%%%%%%%%%%%%%%%%%%%%%%%%%%%%%%%%%%%%%%
\section{Analysis of Mobility Management with Caching Capabilities}
\begin{figure}[!t]
	\centering
	\centerline{\includegraphics[width=5cm]{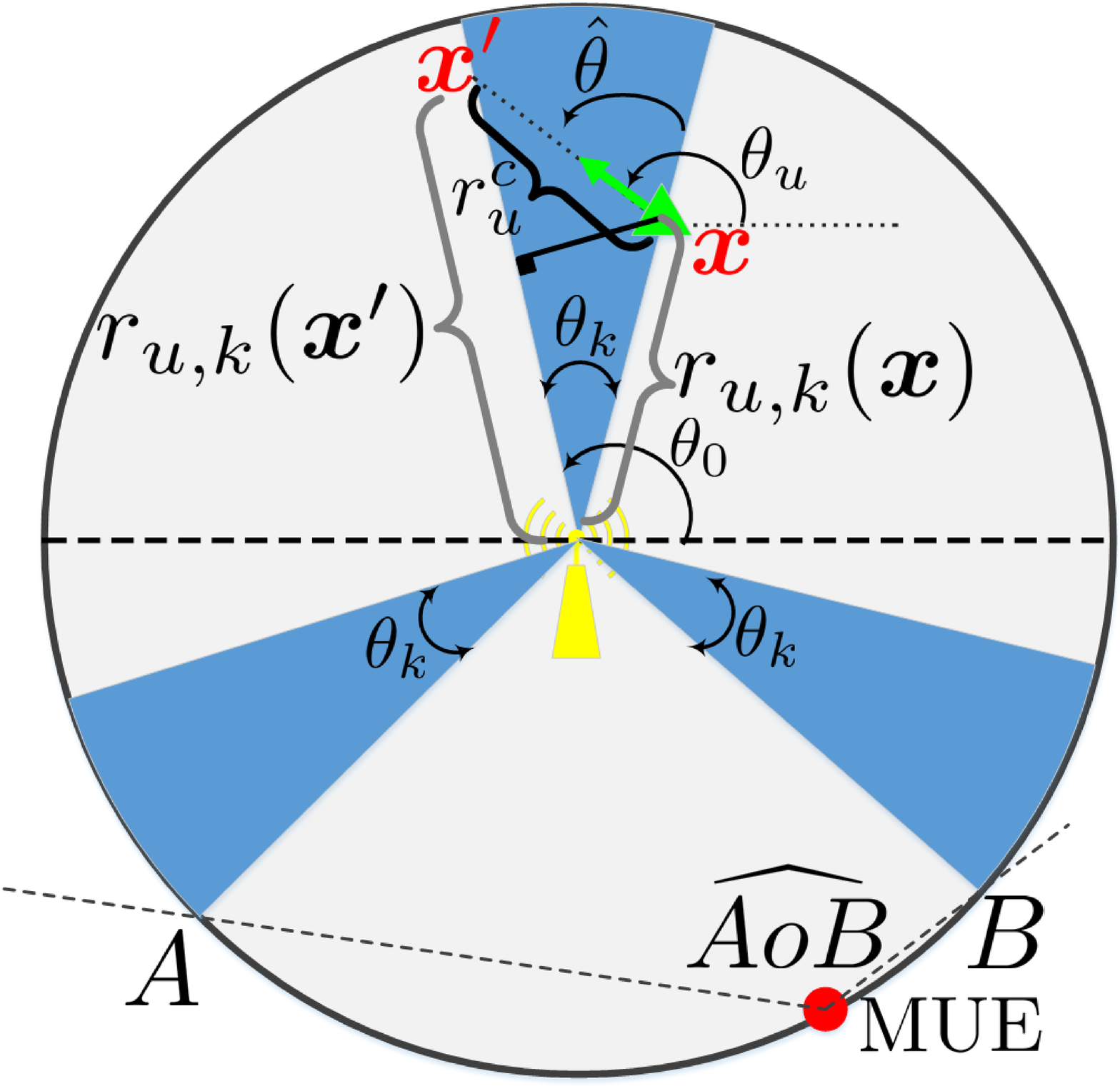}}\vspace{-.2cm}
	\caption{\footnotesize Antenna beam configuration of a dual-mode SBS with $N_k=3$. Blue shaded areas show the mmW beams.}\vspace{-.3cm}
	\label{model2}
\end{figure}\vspace{0em}
First, we investigate the probability of serving an arbitrary MUE over mmW frequencies by a dual-mode SBS. \vspace{-.1cm}
\subsection{Probability of mmW coverage}
In Fig. \ref{model2}, the small circle shows the intersection of an MUE $u$'s trajectory with the coverage area of an arbitrary SBS $k$. In this regard, $\mathbbm{P}^c_{k}(N_k,\theta_k)$ represents the probability that MUE $u$ with a random direction $\theta_u$ and speed $v_u$ crosses the mmW coverage areas of SBS $k$. From Fig. \ref{model2}, we observe that the MUE will pass through the mmW coverage area only if the MUE's direction is inside the angle $\widehat{AoB}$. Hence, we can state the following.

\begin{theorem}\label{prop1}
If the mmW front-end of SBS $k$ has formed $N_k\geq 2$ main lobes, each with a beamwidth $\theta_k>0$, the probability of mmW coverage will be given by:
\begin{align}\label{prop1eq}
\!\!\!\mathbbm{P}^c_{k}(N_k,\theta_k) \!=\!\!\left[\frac{N_k \theta_k}{2\pi}\right] \!\!+\!\! \left[1\!-\!\frac{N_k \theta_k}{2\pi}\right]\!\!\left[\frac{1}{2}\!\left(\!1-\frac{1}{N}\right)\!+\frac{\theta_k}{4\pi}\right].	
\end{align}
\end{theorem}
 \begin{proof}
 See Appendix A.
 \end{proof}
 We can verify \eqref{prop1eq} by considering an example scenario with $N_k=3$ and $\theta_k = \frac{2\pi}{3}$. For this example, \eqref{prop1eq} results in
 $\mathbbm{P}^c_{k}(N_k,\theta_k) =1$ which correctly captures the fact that the entire cell is covered by mmW beams. 
\vspace{-0.2cm}
\subsection{Cumulative distribution function of the caching duration}
To enable an MUE to use the cached content while not being associated to an SBS, it is critical to assess the distribution of caching duration $t^c$ for an arbitrary MUE with a random direction and speed. In this regard, consider the small green triangle in Fig. \ref{model2}, which represents the location of an arbitrary MUE $u$, $\boldsymbol{x_u} = \left(x_u, y_u\right) \in \mathbbm{R}^2$, crossing a mmW beam. First, we note that the geometry of a mmW beam for an SBS can be defined by the location of SBS, as well as the sides of the beam angle. Without loss of generality, assume that the SBS of interest is located at the center, such that $\boldsymbol{x}_k = (0,0)$. Therefore, the two sides of the beam angle will be given by
\begin{align}\label{beamangle1}
y = x\tan(\theta_0-\theta_k), y= x\tan(\theta_0), \,\,\,\, x>0.
\end{align}
Assuming that
the MUE $u$ is currently located on the angle side $x = y\cos(\theta_0-\theta_k)$, as shown by the small triangle in Fig. \ref{model2}, then $\theta_0$ in \eqref{beamangle1} will be $\theta_0 = \arccos\left(\frac{x_u}{r_{u,k}(\boldsymbol{x_u})}\right)+\theta_k$, where $r_{u,k}(\boldsymbol{x})=\sqrt{x_u^2 + y_u^2}$. Hereinafter, we will use the parameter $\theta_0$ to simplify our  analysis. Let $F_{t^c}(.)$ be the cumulative distribution function (CDF) of the caching duration $t^c$. Thus,
\begin{align}\label{caching1}
F_{t_u^c}(t_0) = \mathbbm{P}(t_u^c \leq t_0) = \mathbbm{P}(r_u^c \leq v_u t_0),
\end{align}
where $r_u^c$ denotes the distance that MUE $u$ will traverse across the mmW beam, as shown in Fig. \ref{model2}. Given the location of MUE $\boldsymbol{x}_u$, the minimum possible distance to traverse, $r_u^{\text{min}}$, is
\begin{align}\label{caching2}
r_u^{\text{min}} = \frac{\big|x_u\tan\theta_0 -y_u\big|}{\sqrt{1+ \tan^2\theta_0}}.
\end{align}
In fact, \eqref{caching2} gives the distance of the point $\boldsymbol{x}_u$ from the beam angle side $y = x\tan(\theta_0)$. If $r_u^{\text{min}} >v_u t_0$, then $F_{t_u^c}(t_0) = 0$. Therefore, for the remainder of this analysis we consider $r_u^{\text{min}} \leq v_u t_0$. Next, let $\boldsymbol{x}_u'$ denote the intersection of the MUE's path with line $y = x\tan(\theta_0)$. It is easy to see that $\boldsymbol{x}_u' = \left(x_u + r^c_u \cos\theta_u,y_u + r^c_u\sin\theta_u\right)$. Hence, $y_u + r^c_u\sin\theta_u = \left[x_u + r^c_u \cos\theta_u\right]\tan\theta_0$, and $r^c_u$, i.e., the distance that MUE $u$ traverses during the caching duration $t^c$ is given by: 
\begin{align}\label{caching3}
r^c_u =v_ut_u^c= \frac{y_u-x_u\tan\theta_0}{\tan\theta_0 \cos \theta_u-\sin \theta_u}.
\end{align}
Next, from \eqref{caching1} and \eqref{caching3}, the CDF can be written as  
\begin{align}\label{caching4}
F_{t_u^c}(t_0)=\mathbbm{P}\left(\frac{y_u-x_u\tan\theta_0}{\tan\theta_0 \cos \theta_u-\sin \theta_u} \leq v_u t_0\right).
\end{align}
 Using the geometry shown in Fig. \ref{model2}, we find the CDF of caching duration as follows:
 \begin{lemma}\label{prop2}
 The CDF of caching duration, $t^c$, for an arbitrary MUE $u$ with speed $v_u$ is given by
 \begin{align}\label{caching5}
 F_{t^c}(t_0)&= \frac{1}{\pi-\theta_k}\bigg(\arccos\left(\frac{r_u^{\text{min}}}{v_ut_0}\right)\\\notag
 &+\min\bigg\lbrace \arccos\left(\frac{r_u^{\text{min}}}{r_{u,k}(\boldsymbol{x})}\right), \arccos\left(\frac{r_u^{\text{min}}}{v_ut_0}\right)\bigg\rbrace\bigg).
 \end{align}
 \end{lemma}
 \begin{proof}
 	 See Appendix B.
 \end{proof}
 The CDF of $t^c$ is shown in Fig. \ref{CDF1} for different MUE distances from the serving SBS. Fig. \ref{CDF1} shows that as the MUE is closer to the SBS, $t^c$ takes smaller values with higher probability which is expected, since the MUE will traverse a shorter distance to cross the mmW beam.
\section{Performance Analysis of the Proposed Caching-enabled Mobility Management Scheme}
\begin{figure}[!t]
	\centering
	\centerline{\includegraphics[width=8cm]{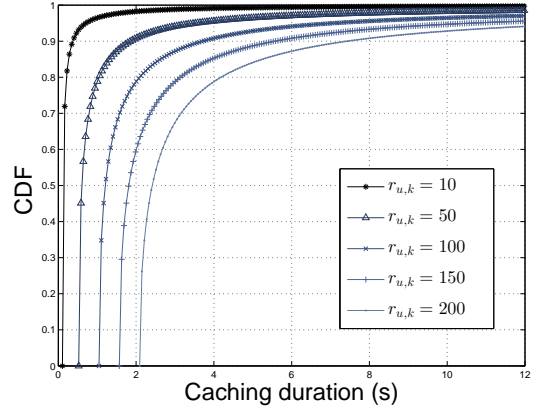}}\vspace{-.2cm}
	\caption{\footnotesize CDF of caching duration $t^c$.}\vspace{-0.3cm}
	\label{CDF1}
\end{figure}
Next, we analyze the average achievable rate for content caching, for an MUE with speed $v_u$, direction $\theta_u$, and initial distance $r_{u,k}(\boldsymbol{x})$ from the serving dual-mode SBS. In addition, we evaluate the impact of caching on mobility management. 
%For this analysis, we ignore the shadowing effect and only consider distance path loss. 
\subsection{Average Achievable Rate for Caching}
The achievable rate of caching follows
\begin{align}\label{rate1}
	\!\!R^c(u,k) = \frac{1}{v_ut_u^c}\int_{r_{u,k}(\boldsymbol{x})}^{r_{u,k}(\boldsymbol{x}')}\!\!\!\!\!\!w\log\left(1+ \frac{\beta P_t\psi r_{u,k}^{-\alpha}}{wN_0}\right)dr_{u,k},
\end{align}
where $\beta =(\frac{\lambda}{4\pi r_0})^2r_0^{\alpha}$. The integral in \eqref{rate1} is taken over the line with length $r_u^c$ that connects the MUE location  $\boldsymbol{x}$ to $\boldsymbol{x}'$, as shown in Fig. \ref{model2}. 
\begin{theorem}
	The average achievable rate for an MUE $u$ served by an SBS $k$, $\bar{R}^c(u,k)$, is given by 
	\begin{align}\label{rate2}
		\!\!\!\!\!\!	&\bar{R}^c(u,k) = \mathbbm{P}^c_{k}(N_k,\theta_k)R^c(u,k),\\\label{rate18}
			&=\delta_2 \int_{f(\theta_k)}^{f(0)}\frac{1}{f^2(\theta)}\log\left(1+\delta_1f^{\alpha}(\theta)\right)df(\theta),\\\notag
			& \stackrel{\text{(a)}}{=} \frac{\delta_2}{\ln(2)}\! \bigg[2\sqrt{\delta_1}\arctan(\sqrt{\delta_1}f(\theta_k))\!-\!\frac{\ln(\delta_1f^2(\theta_k)+1)}{f(\theta_k)},\\\label{rate19}
			&-\! 2\sqrt{\delta_1}\arctan(\sqrt{\delta_1}f(0))\!+\!\frac{\ln(\delta_1f^2(0)+1)}{f(0)}\bigg],
	\end{align}
	where $\delta_1 = \beta P_t\psi(r_{u,k}(\boldsymbol{x})\sin\hat{\theta})^{-\alpha}/wN_0$. Moreover, $\delta_2=wr_{u,k}(\boldsymbol{x})\sin\hat{\theta}\mathbbm{P}^c_{k}(N_k,\theta_k)/v_ut^c$, and $\hat{\theta} = \theta_u-\theta_0+\theta_k$. For (a) to hold, we set  $\alpha=2$ which is a typical value for the path loss exponent of LoS mmW links \cite{Ghosh14}.
\end{theorem}
\begin{proof}
	Theorem \ref{prop1} implies that with probability $1-\mathbbm{P}^c_{k}(N_k,\theta_k)$, only $\mu$W coverage is available for an MUE. Therefore, average achievable rate for caching over the mmW frequencies is given by \eqref{rate2}. To simplify \eqref{rate2}, we have
	\begin{align}\label{rate20}
		r_{u,k}\cos\theta = r_{u,k}(\boldsymbol{x}) + r_u\cos\hat{\theta}, \,\, 	r_{u,k}\sin\theta = r_u\sin\hat{\theta},
	\end{align}
	where $\hat{\theta} = \theta_u-\theta_0+\theta_k$ and $\theta$ is an angle between the line connecting MUE to SBS, ranging from $0$ to $\theta_k$. Moreover, $r_u$ is the current traversed distance, with $r_u = r_u^c$ once the MUE reaches $\boldsymbol{x}'$ by the end of caching duration, as shown in Fig. \ref{model2}. From \eqref{rate20}, we find $r_{u,k} = r_{u,k}(\boldsymbol{x})\sin\hat{\theta}/\sin(\hat{\theta}-\theta)$. By changing the integral variable $r_u$ to $\theta$, we can write \eqref{rate2} as 
	\begin{align}\label{rate21}
		\!\!\!\!\!\bar{R}^c(u,k) \!=\! \delta_2\!\!\int_{0}^{\theta_k}\!\!\log\!\left(1\!+\! \delta_1\sin^{\alpha}(\hat{\theta}-\theta)\right)\frac{\cos(\hat{\theta}-\theta)}{\sin^2(\hat{\theta}-\theta)}d\theta,\vspace{-.1cm}
	\end{align}
	where $\delta_1 = \beta P_t\psi(r_{u,k}(\boldsymbol{x})\sin\hat{\theta})^{-\alpha}/wN_0$ and $\delta_2=wr_{u,k}(\boldsymbol{x})\sin\hat{\theta}\mathbbm{P}^c_{k}(N_k,\theta_k)/v_ut^c$. Next, we can directly conclude \eqref{rate18} from \eqref{rate21} by substituting $f(\theta) = \sin(\hat{\theta}-\theta)$ in \eqref{rate21}. For $\alpha=2$, which is a typical value for the path loss exponent for LoS mmW links, \eqref{rate18} can be simplified into \eqref{rate19} by taking the integration by parts in \eqref{rate18}.
\end{proof}
\subsection{Achievable gains of caching for mobility management}
From \eqref{traffic1}, \eqref{traffic2}, and \eqref{rate19}, we can find $d^c(u,k)$ which is the distance that MUE $u$ can traverse, while using the cached video content. On the other hand, by having the average inter-cell distances in a HetNet, we can approximate the number of SBSs that an MUE can pass over distance $d^c(u,k)$. Hence, the average number of SBSs that MUE is able to traverse without performing cell search for HO is
\begin{align}\label{perf1}
\zeta\approx\bigg\lfloor \frac{\mathbbm{E}\left[d^c(u,k)\right]}{l}\bigg\rfloor, 
\end{align}
where the expected value is taken, since $d^c(u,k)$ is a random variable that depends on  $\theta_u$. Moreover, $l$ denotes the average inter-cell distance. Here, we note that 
\begin{align}\label{perf2}
\mathbbm{E}\left[d^c(u,k)\right] = \int_{0}^{\infty}(1-F_{t_u^c}(v_ut))dt,	
\end{align}
where $F_{t^c}(.)$ is derived in Lemma \ref{prop2}. We note that \eqref{perf2} is the direct result of writing an expected value in terms of CDF. Based on the definition of $\zeta$ in \eqref{perf1}, we can conclude the following.  
\begin{remark}
	The proposed caching scheme will reduce the average number of HOs by $1/\zeta$ factor.
\end{remark}
Furthermore, from the definition of HOF $\gamma_{\text{HOF}}$ in \eqref{HOF1}, we can define the probability of HOF as $\mathbbm{P}(D_{u,k}<v_ut_{\text{MTS}})$ \cite{6849322}, where $D_{u,k}=t_{u,k}/v_u$, and $t_{u,k}$ is the ToS. To compute the HOF probability, we use the probability density function (PDF) of a random chord length within a circle with radius $a$, as follows:\vspace{-.1cm}
\begin{align}\label{perf3}
	f_D(D)=\frac{2}{\pi\sqrt{4a^2-D^2}},
\end{align}
where \eqref{perf3} relies on the assumption that one side of the chord is fixed and the other side is determined by choosing a random $\theta \in \left[0,\pi\right]$. This assumption is in line with our analysis as shown in Fig. \ref{model2}. Using \eqref{perf3}, we can find the probability of HOF as follows:
\begin{align}
	\mathbbm{P}(D_{u,k}<v_ut_{\text{MTS}})&=\int_{0}^{v_ut_{\text{MTS}}} \frac{2}{\pi\sqrt{4a^2-D^2}}dD,\\
	&=\frac{2}{\pi}\arcsin\left(\frac{v_ut_{\text{MTS}}}{2a}\right).
	\end{align}
In fact, $\gamma_{\text{HOF}}$ is a Bernoulli random variable with a probability of success that depends on the MUE's speed, cell radius, and $t_{\text{MTS}}$. Hence, by reducing the number of HOs by $1/\zeta$ factor, the proposed scheme will reduce the expected value of the sum $\sum \gamma_{\text{HOF}}$, taken over all SBSs that an MUE visits during the considered time $T$. 
%Given the random distribution of SBSs, we can assume that $\gamma_{\text{HOF}}$ of one SBS is independent of other SBSs. Thus, the average HOF increases as an MUE traverses more SBSs. However, by leveraging caching, average HOF 
\vspace{-0cm}
\section{Simulation Results}
\vspace{-0cm}
For simulations, we consider a HetNet composed of $K=50$ SBSs distributed uniformly across a circular area with radius $500$ meters with the MBS located at the center and a minimum inter-cell distance of $30$ meters. The main parameters are summarized in Table \ref{tab1}. In our simulations, we consider the overall transmit-receive antenna gain from an interference link to be random. All statistical results are averaged over a large number of independent runs. 
\begin{table}[t!]
	\scriptsize
	\centering
	\caption{%\mycaption{%\vspace*{-1em}
		\vspace*{-0em}Simulation parameters}\vspace*{-1em}
	\begin{tabular}{|c|c|c|}
		\hline
		\bf{Notation} & \bf{Parameter} & \bf{Value} \\
		\hline
		$f_c$ & Carrier frequency & $73$ GHz\\
		\hline
		$P_{t,k}$ & Total transmit power of SBSs & $\left[20, 27, 30\right]$ dBm\\
		\hline
		$K$ & Total number of SBSs & $50$\\
		\hline
		$w$ & Available Bandwidth & $5$ GHz\\
		%\hline
		%($\xi_{\text{LoS}}$,$\xi_{\text{NLoS}}$) & Standard deviation of mmW path loss& ($4.2, 7.9$) \cite{Ghosh14} \\
		\hline
		($\alpha_{\text{LoS}}$,$\alpha_{\text{NLoS}}$) & Path loss exponent& ($2,3.5$) \cite{Ghosh14}\\
		\hline
		$d_0$ & Path loss reference distance& $1$ m \cite{Ghosh14}\\
		\hline
		$G_{\text{max}}$ & Antenna main lobe gain& $18$ dB \\
		\hline
		$G_{\text{min}}$ & Antenna side lobe gain& $-2$ dB \\
		\hline
		$N_k$ & Number of mmW beams& $3$ \\
		\hline
		$\theta_{m}, \theta_k$ & beam width& $10^{\circ}$ \\
		\hline
		$N_0$ & Noise power spectral density& $-174$ dBm/Hz\\
		\hline
		$t_{\text{MTS}}$ & Minimum time-of-stay& $1$s \cite{3gpp} \\
		\hline
		$Q$ & Play rate& $1$k segments per second \\
		\hline
		$B$ & Size of video segments& $1$ Mbits  \\
		\hline
		$v_u$ & MUE speed& $\left[3, 10, 30, 45, 60\right]$ km/h  \\
		\hline
	\end{tabular}\label{tab1}\vspace{-0.4cm}
\end{table}
%%%%%%%%%%%%%%%%%%%%%%%%%%%%%%%%%%%%%%%%%%%%%%%%%%%%%%%%%%%%%%%%%%%%%%%%%%%%%%%%%%%

Fig. \ref{main1sim} compares the average HOF of the proposed scheme with a conventional HO mechanism that relies on the average RSS to perform HO and does not exploit caching. The results clearly demonstrate that caching capabilities, as proposed here, will significantly improve the HO process for dense HetNets. In fact, the results in Fig. \ref{main1sim} show that caching over mmW frequencies will reduce HOF for all speed, reaching up to $45 \%$ for MUEs with $v_u=60$ km/h. 
 \begin{figure}[!t]
 	\centering
 	\centerline{\includegraphics[width=9cm]{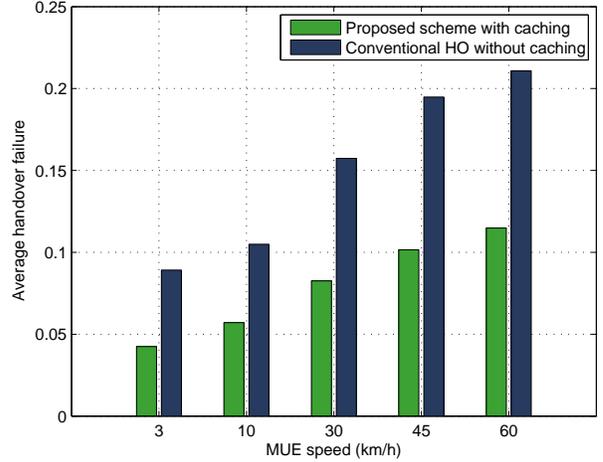}}\vspace{-.2cm}
 	\caption{\footnotesize HOF vs different MUE speeds.}\vspace{-0.3cm}
 	\label{main1sim}
 \end{figure}\vspace{0em}
  
Furthermore, Fig. \ref{rate} shows the achievable rate of caching for an MUE with $v_u=60$ km/h, as function of different initial distances $r_{u,k}(\boldsymbol{x})$ for various $\theta_u$. The results in Fig. \ref{rate} show that even for MUEs with high speeds, the achievable rate of caching is significant, exceeding $10$ Gbps, for all $\theta_u$ values and inital distance of $20$ meters from the SBS. However, we can observe that the blockage can noticeably degrade the performance. In fact, for NLoS scenarios,  the maximum achievable rate at the distance of $20$ meters reduces to $2$ Gbps.\vspace{-0cm}
 \section{Conclusions}\label{conclude}\vspace{-0cm}
In this paper, we have proposed a comprehensive framework for mobility management in integrated microwave-millimeter wave networks. In particular, we have shown that by smartly caching video contents while exploiting the dual-mode nature of the network's base stations, one can provide seamless mobility to the users. We have derived various fundamental results on the probability and the achievable rate for caching video contents by leveraging millimeter wave high capacity transmissions. We have shown that caching provides significant gains in reducing the number of handovers. Numerical results have corroborated our analytical results and showed that the significant rates for caching can be achieved over the mmW frequencies, even for fast mobile users. The results also have shown that the proposed approach substantially decreases the handover failures in heterogeneous networks.

 \begin{figure}[!t]
 	\centering
 	\centerline{\includegraphics[width=9cm]{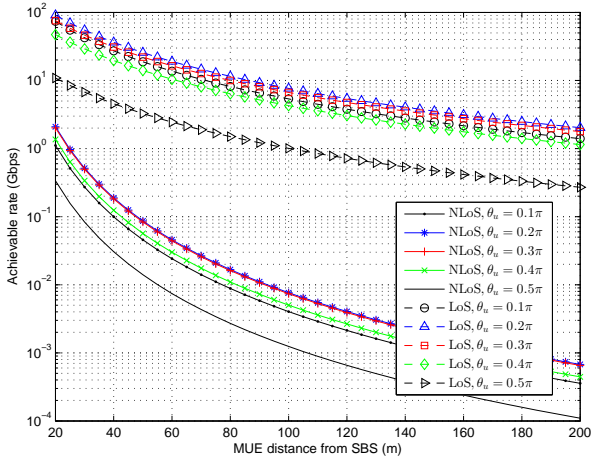}}\vspace{-.4cm}
 	\caption{\footnotesize Achievable rate of caching vs $r_{u,k}(\boldsymbol{x})$ for different $\theta_u$.}\vspace{-.3cm}
 	\label{rate}
 \end{figure}\vspace{0em}
\vspace{-0.15cm} 
 \section*{Appendix A}\label{app1}\vspace{-0.15cm} 
 Due to the equidistant beams, we have 
 \begin{align}\label{app1-1}
 	\widehat{AoB} &= \frac{1}{2}AB =\frac{1}{2}\left[2\pi-AoB\right]\\\notag
 	 &= \frac{1}{2}\left[2\pi - \left(\frac{2\pi}{N_k}-\theta_k\right)\right]
 =\left(1-\frac{1}{N_k}\right)\pi + \frac{\theta_k}{2}.
 \end{align}
 Given that an arbitrary MUE can enter the circle in Fig. \ref{model2} with any direction, it will be instantly covered by mmW with probability $\mathbbm{P}(\boldsymbol{x}_u \in \mathcal{A}) \!=\! \frac{N_k\theta_k}{2\pi}$, where $\mathcal{A} \!\subset \!\mathbbm{R}^2$ denotes the part of circle's perimeter that overlaps with mmW beams. Thus, \vspace{-.2cm}
 \begin{align}\label{app1-2}
 \!\!\!	\mathbbm{P}^c_{k}(N_k,\theta_k) = \mathbbm{P}(\boldsymbol{x}_u \in \mathcal{A}) + \left[1-\mathbbm{P}(\boldsymbol{x}_u \in \mathcal{A})\right]\frac{1}{2\pi}\widehat{AoB},
 \end{align}
where \eqref{app1-2} is resulted from the fact that $\theta_u \sim U\left[0, 2\pi\right]$. Therefore, from \eqref{app1-1} and \eqref{app1-2}, the probability of crossing a mmW beam follows \eqref{prop1eq}.
\vspace{-.3cm}
 \section*{Appendix B}\label{app2}
 From \eqref{caching1}, $F_{t^c}(t_0)=\mathbbm{P}(r_u^c\leq v_u t_0)$. To find this probability, we note that $r_u^c\leq v_u t_0$ if MUE moves between two line segments of length $v_ut_0$ that connect MUE to line $y=x\cos\theta_0$. Depending on $r_{u,k}(\boldsymbol{x})$, the intersection of line segment with $y=x\cos\theta_0$ may have one or two solutions. In case of two intersection points, the two line segments will make two equal angles with the perpendicular line from $\boldsymbol{x}_u$, to $y=x\cos\theta_0$, which each is obviously equal to $\pi-(\pi/2-\theta_k)-\hat{\theta}=\pi/2+\theta_k-\hat{\theta}=\arccos\left(\frac{r_u^{\text{min}}}{v_ut_0}\right)$. Therefore, 
 \begin{align}\label{aaaeq1}
 	F_{t^c}(t_0)=\frac{2}{\pi-\theta_k}\arccos\left(\frac{r_u^{\text{min}}}{v_ut_0}\right).
 \end{align}
 In fact, $\theta_u$ must be within a range of $\pi-\theta_k$ for $r_u^c\leq v_ut_0$ to be valid. 
 Now, if this angle is greater than $\pi/2-\theta_k$, only one intersection point exists. Equivalently,
  \begin{align}\label{aaaeq2}
  \!\!\!\!F_{t^c}(t_0)\!=\!\frac{1}{\pi-\theta_k}\!\bigg(\!\!\arccos\left(\frac{r_u^{\text{min}}}{v_ut_0}\right)\!+\!\arccos\left(\frac{r_u^{\text{min}}}{r_{u,k}(\boldsymbol{x})}\right)\!\!\bigg).
  \end{align}
 Integrating \eqref{aaaeq1} and \eqref{aaaeq2}, the CDF for caching duration can be written as \eqref{caching5}.\vspace{-.1cm} 

\vspace{-0em}
\def\baselinestretch{0.92}
\bibliographystyle{ieeetr}
\bibliography{references}
\end{document}